\documentclass[conference,compsoc,10pt]{IEEEtran}
\usepackage{bm}        
\usepackage{amssymb,amsmath,amsthm}   
\usepackage{nicefrac}
\usepackage[PostScript=dvips]{diagrams}

\newtheorem{theorem}             {Theorem}
\newtheorem{proposition}[theorem]{Proposition}

\newtheorem{corollary}  [theorem]{Corollary}

\newarrow{Eq}=====
\newarrow{mon}>--->
\newarrow{epi}----{>>}
\newarrow{mto}|--->
\newarrow{rel}--+->
\newarrow{inc}C--->
\newarrow{pfn}----{harpoon}

\diagramstyle[newarrowhead=vee,newarrowtail=vee]

\newcommand{\Sig}{\sigma}
\newcommand{\CS}{\mathcal{R}(\Sig)}
\newcommand{\CSf}{\mathcal{R}_{\fin}(\Sig)}
\newcommand{\fin}{f}
\newcommand{\Lk}{\mathcal{L}^{k}_{\infty,\omega}}
\newcommand{\lk}{\exists^{+}\mathcal{L}^{k}_{\omega,\omega}}
\newcommand{\LkEP}{\exists^{+}\mathcal{L}^{k}_{\infty,\omega}}
\newcommand{\rhom}{\rightarrow}
\newcommand{\Tk}{\mathbb{T}_k}

\newcommand{\Tl}{\mathbb{T}_l}

\newcommand{\pref}{\sqsubseteq}
\newcommand{\epA}{\varepsilon_{A}}
\newcommand{\epB}{\varepsilon_{B}}
\newcommand{\epT}{\varepsilon_{\Tk A}}
\newcommand{\id}{\mathsf{id}}
\newcommand{\CK}{\mathcal{K}(\Tk)}
\newcommand{\Ktwo}{\mathcal{K}(\mathbb{T}_{2})}
\newcommand{\ie}{\textit{i.e.}~}
\newcommand{\form}{\varphi}
\newcommand{\Log}{\mathcal{L}}
\newcommand{\eqL}{\equiv^{\Log}}
\newcommand{\ifdef}{\;\; \stackrel{\Delta}{\Longleftrightarrow} \;\;}
\newcommand{\IFF}{\;\; \Longleftrightarrow \;\;}
\newcommand{\ifaoif}{\; \Longleftrightarrow \;}
\newcommand{\rarr}{\rightarrow}
\newcommand{\arrk}{\rightarrow_{k}}
\newcommand{\arrl}{\rightarrow_{l}}
\newcommand{\arrki}{\rightarrow_{k}^{i}}
\newcommand{\arrks}{\rightarrow_{k}^{s}}
\newcommand{\arrkb}{\rightarrow_{k}^{b}}

\newcommand{\arrkd}{\rightarrow_{k}^{\mathsf{d}}}

\newcommand{\preford}{\sqsubseteq}

\newcommand{\rlk}{\rightleftarrows_{k}}
\newcommand{\rli}{\rightleftarrows_{k}^{i}}
\newcommand{\rls}{\rightleftarrows_{k}^{s}}
\newcommand{\rlkm}{\rightleftarrows_{k}^{-}}
\newcommand{\isoK}{\cong_{\mathcal{K}}}
\newcommand{\eqk}{\eposL}

\newcommand{\eposL}{\equiv^{\exists^{+}}}
\newcommand{\gplay}{[ (p_1, a_1), \ldots , (p_n, a_n)]}
\newcommand{\posn}{\mathsf{position}}
\newcommand{\va}{\mathbf{a}}
\newcommand{\vb}{\mathbf{b}}
\newcommand{\vc}{\mathbf{c}}
\newcommand{\vx}{\mathbf{x}}
\newcommand{\vy}{\mathbf{y}}
\newcommand{\vz}{\mathbf{z}}
\newcommand{\vs}{\mathbf{s}}
\newcommand{\vt}{\mathbf{t}}
\newcommand{\vu}{\mathbf{u}}

\newcommand{\gm}{\gamma}
\newcommand{\Gmk}{\Gamma_{k}}
\newcommand{\Gf}{S_f}
\newcommand{\thsp}{\psi_{s,p}}

\newcommand{\tha}{\theta_{\gamma}}

\newcommand{\dom}{\mathrm{dom}}
\newcommand{\emp}{\varnothing}
\newcommand{\Sr}{S_{d}}
\newcommand{\rel}{R}
\newcommand{\pb}{\pi}
\newcommand{\tw}{\mathsf{tw}}
\newcommand{\core}{\mathsf{core}}
\newcommand{\card}{\mathsf{card}}
\newcommand{\Bk}{\Box_{k}}
\newcommand{\Dh}{\Diamond}
\newcommand{\Bl}{\Box_{l}}
\newcommand{\vphi}{\varphi}
\newcommand{\imp}{\; \supset \;}
\newcommand{\vsn}{\vspace{-0.075in}}
\newcommand{\prd}{\prec}
\newcommand{\cn}{\kappa}
\newcommand{\scB}{\mathsf{sc}_{B}}
\newcommand{\CSP}{\mathrm{CSP}}
\newcommand{\ZZ}{\mathbb{Z}}
\newcommand{\tf}{\theta_{f}}

\newcommand{\rev}[1]{#1^{-}}

\begin{document}

\title{The Pebbling Comonad in Finite Model Theory}

\author{\IEEEauthorblockN{Samson Abramsky}
\IEEEauthorblockA{Department of Computer Science, University of Oxford, Oxford, U.K.\\
Email: samson.abramsky@cs.ox.ac.uk}
\IEEEauthorblockN{Anuj Dawar and Pengming Wang}
\IEEEauthorblockA{Computer Laboratory, University of Cambridge}
Email: anuj.dawar@cl.cam.ac.uk, pengming.wang@cl.cam.ac.uk
}
\date{}

\IEEEoverridecommandlockouts
\IEEEpubid{\makebox[\columnwidth]{978-1-5090-3018-7/17/\$31.00~
\copyright2017 IEEE \hfill} \hspace{\columnsep}\makebox[\columnwidth]{ }}

\maketitle

\begin{abstract}
Pebble games are a powerful tool in the study of finite model theory, constraint satisfaction and database theory. Monads and comonads are basic notions of category theory which are widely used in semantics of computation and in modern functional programming.  We show that existential k-pebble games have a natural comonadic formulation.
Winning strategies for Duplicator in the k-pebble game for structures A and B are equivalent to morphisms from A to B in the coKleisli category for this comonad. This leads on to comonadic characterisations of a number of central concepts in Finite Model Theory:
\begin{itemize}
\item Isomorphism in the co-Kleisli category characterises elementary equivalence in the k-variable logic with counting quantifiers.
\item Symmetric games corresponding to equivalence in full k-variable logic are also characterized.
\item The treewidth of a structure A is characterised in terms of its coalgebra number: the least k for which  there is a coalgebra structure on A for the k-pebbling comonad.
\item Co-Kleisli morphisms are used to characterize strong consistency, and to give an account of a Cai-F\"urer-Immerman construction.
\item The k-pebbling comonad is also used to give semantics to a novel modal operator.
\end{itemize}
These results lay the basis for some new and promising connections between two areas within logic in computer science which have  largely been disjoint: (1) finite and algorithmic model theory, and (2) semantics and categorical structures of computation.

\end{abstract}

\section{Introduction}

Homomorphisms play a fundamental r\^ole in finite model theory, constraint satisfaction and database theory. The existence of a homomorphism $A \rhom B$ is an equivalent formulation of the basic CSP problem \cite{libkin2013elements,chandra1977optimal,feder1998computational}. There is an equivalence between the existence of a homomorphism, and the property that every existential positive sentence satisfied by $A$ is also satisfied by $B$ \cite{chandra1977optimal}.
Such sentences correspond to (disjunctions of) \emph{conjunctive queries}, which are fundamental in database theory \cite{abiteboul1995foundations,kolaitis1998conjunctive}.

One of the key tools in studying these notions is that of {existential $k$-pebble games} \cite{kolaitis1990expressive}. Such a game, for structures $A$, $B$, proceeds by Spoiler placing one of his $k$ pebbles on an element of the universe of $A$. Duplicator then places one of her pebbles on an element of $B$. If Duplicator is always able to move so that the partial mapping from $A$ to $B$ defined by sending $a_i$, the element in $A$ carrying the $i$'th Spoiler pebble, to $b_i$, the corresponding element of $B$ carrying the $i$'th Duplicator pebble, is a homomorphism on the induced substructures, then Duplicator has a winning strategy. 
\begin{proposition}[\cite{kolaitis1990expressive}]
\label{pebbleposprop}
The following are equivalent:
\begin{itemize}
\item Duplicator has a winning strategy in the existential $k$-pebble game.
\item Every sentence of the existential positive $k$-variable fragment of first-order logic satisfied by $A$ is also satisfied by $B$.
\end{itemize}
\end{proposition}

Our aim in this paper is to study these notions  from a novel perspective, using the notion of \emph{comonad} from category theory.
Monads and comonads are basic notions of category theory which are widely used in semantics of computation and in modern functional programming 
\cite{moggi1991notions,brookes1991computational,wadler1995monads}.  We show that existential $k$-pebble games have a natural comonadic formulation.
Given a structure $A$ over a relational signature $\Sig$, we shall  introduce a new structure $\Tk A$ corresponding to Spoiler playing his part of an existential  $k$-pebble game on $A$, with the potential codomain $B$ left unspecified. The idea is that we can exactly recover the content of a Duplicator strategy in $B$ by giving a homomorphism from $\Tk A$ to $B$. Thus the notion of \emph{local approximation} built into the $k$-pebble game is internalised into the category of $\Sig$-structures and homomorphisms. Formally, this construction will be shown to give a \emph{comonad} on this category, which guarantees a wealth of further structural properties.
This leads  to comonadic characterisations of a number of central concepts in Finite Model Theory.

In Section~2, we introduce the pebbling comonads $\Tk$, which are graded by the number of pebbles $k$, and characterize their coalgebras.
$\Tk A$ is always infinite. In Section~3, we prove a no-go theorem, to rule out any finite version.
In Section~4, we show that the question of whether a morphism $\Tk A \rarr B$ exists is equivalent to the existence of a winning strategy for the existential $k$-pebble game for $A$ and $B$.
We also show how this question can be finitized, and such morphisms
can be realized by deterministic finite-state transducers. In
Section~5 we study various notions of equivalence of structures, and
show that isomorphism in the coKleisli category for $\Tk$ coincides
with elementary equivalence for $C^k$, the $k$-variable logic with
counting quantifiers, which plays a central r\^ole in finite model
theory. We also show that the equivalence of structures $A$ and $B$ for full $k$-variable logic is characterized by the existence of a pair of morphisms $\Tk A \rarr B$ and $\Tk B \rarr A$ satisfying a certain condition. This is important, as it shows that symmetric, back-and-forth game conditions also fall within the scope of the comonadic approach.
In Section~6,
we characterise treewidth, a combinatorial parameter which plays a
pervasive r\^ole in algorithmic graph theory, in terms of the coalgebra number of a structure $A$: the least index $k$ such that there is a coalgebra $A \rarr \Tk A$. Section~7 characterizes the weaker condition of existence of a homomorphism from $A$ to $\Tk A$ in terms of the treewidth of the core of $A$. Section~8 discusses strong $k$-consistency,  and gives a Cai-F\"urer-Immerman construction in terms of coKleisli morphisms. In Section~9, we introduce the novel modality $\Bk$, which can be read as ``it is $k$-locally the case that'', and give its semantics using the pebbling comonad. Section~10 concludes.

Overall, these results lay the basis for some new and promising connections between two areas within logic in computer science which have  largely been disjoint: 
(1) finite and algorithmic model theory, and (2) semantics and categorical structures of computation.

\subsection*{Notation and Background}
We have attempted to make this paper accessible to non-specialists in category theory. We  assume only very basic background in category theory (see e.g. \cite{pierce1988taste,abramsky2010introduction}), essentially the definitions of category and functor, spelling out other definitions as needed.

We also assume only basic  background in (finite) model theory (see
e.g. \cite{libkin2013elements}). We use the letters $A$, $B$ for
$\Sig$-structures over some relational signature $\Sig$, not
distinguishing notationally between the structures and their
universes. The interpretation of a relation $R \in \Sig$ in a
structure $A$ is denoted by $R^A$. Our morphisms between structures will always be homomorphisms. 
We use $\Lk$ to denote the $k$-variable fragment of first-order logic extended with infinite disjunctions and conjunctions, $\LkEP$ for its existential positive fragment (no universal quantifiers or negations), and $\lk$ for the corresponding fragment of first-order logic.

We use the notation $[n] := \{ 1, \ldots , n\}$, and $s \preford t$ for the prefix ordering on sequences. If $s \preford t$, there is a unique $s'$ such that $ss' = t$, which we refer to as the suffix of $s$ in $t$.

\section{The Pebbling Comonad}

For the remainder of the paper, we fix a finite relational signature $\Sig$. Our main setting is $\CS$, the category whose objects are  $\Sig$-structures, and whose morphisms are homomorphisms of $\Sig$-structures. We shall work in $\CS$ and various sub-categories thereof,  in particular $\CSf$, the category of finite $\Sig$-structures. By ``structure'', we always mean $\Sig$-structure.

Firstly, we can consider the set of all plays in $A$ by the Spoiler. This can be represented by $([k] \times A)^{+}$, the set of finite non-empty sequences of moves $(p, a)$, where $p \in [k]$ is a pebble index, and $a \in A$. We shall use the notation $s = [ (p_1, a_1), \ldots , (p_n, a_n)]$ for these sequences.
The set of plays forms the universe of $\Tk A$. 

To complete the definition of $\Tk A$, we must define the relational structure on this universe. For simplicity, we first consider the case of a binary relation $E$. We define $E^{\Tk A}$ to be the set of pairs of plays $s, t \in \Tk A$ such that 
\begin{itemize}
\item $s$ and $t$ are comparable in the prefix ordering, so $s \pref t$ or $t \pref s$.
\item If $s \preford t$, then the pebble index of the last move in $s$ does not appear in the suffix of $s$ in $t$; and symmetrically if $t \preford s$.
\item $E^{A}(\epA(s), \epA(t))$, where $\epA : \Tk A \rTo A$ sends a play  $\gplay$ to $a_n$, the $A$-component of its last move.
\end{itemize}
To understand the second condition in this definition, note that  for each pebble index $p$, only the \emph{last} move with pebble index $p$ is relevant to the current position. The idea is that, in placing pebble $p$ on element $a$, Spoiler has to first remove it from its previous position. This is exactly the way in which the pebble game models bounded resources. 

The extension to relations of arbitrary arity is straightforward. Given an $m$-ary relation $R \in \Sig$, $R^{\Tk A}(s_1, \ldots , s_m)$ if the $s_i$ are pairwise comparable in the prefix ordering, and hence form a chain with greatest element $s$; the pebble index of the last move in each $s_i$ does not appear in the suffix of $s_i$ in $s$; and $R^{A}(\epA(s_1), \ldots , \epA(s_m))$.

The following is immediate from our definition.
\begin{proposition}
The map $\epA : \Tk A \rTo A$ is a homomorphism.
\end{proposition}
We now extend $\Tk$ to a functor $\Tk : \CS \rTo \CS$ by defining its action on morphisms. If $f : A \rTo B$ is a homomorphism, we define $\Tk f : \Tk A \rTo \Tk B$ to be the map
\[  \gplay \; \mapsto \;  [ (p_1, f(a_1)), \ldots , (p_n, f(a_n))] . \]
It is clear that this is a homomorphism from $\Tk A$ to $\Tk B$. Moreover, it is easily verified that $\Tk(g \circ f) = \Tk(g) \circ \Tk(f)$, and $\Tk( \id_A) = \id_{\Tk A}$, so $\Tk$ is a functor.

We have already defined the homomorphism $\epA : \Tk A \rTo A$ for each structure $A$. We now note that this defines a natural transformation. That is, for each homomorphism $f : A \rTo B$, the following diagram commutes.
\begin{diagram}
\Tk A & \rTo^{\epA} & A \\
\dTo^{\Tk f} & & \dTo_{f} \\
\Tk B & \rTo_{\epB} & B
\end{diagram}
Now for each structure $A$, we define a map $\delta_A : \Tk A \rTo \Tk \Tk A$, as follows. Given a play $s = \gplay$, define $s_i = [ (p_1, a_1), \ldots , (p_i, a_i) ]$, $i = 1, \ldots , n$. Then we define
\[ \delta_A : s \mapsto [ (p_1, s_1) ,  \ldots , (p_n, s_n) ] . \]
\begin{proposition}
For each $A$, $\delta_A$ is a homomorphism. Moreover, $\delta$ is a natural transformation: for each homomorphism $f : A \rTo B$, the following diagram commutes.
\begin{diagram}
\Tk A & \rTo^{\delta_A} & \Tk \Tk A \\
\dTo^{\Tk f} & & \dTo_{\Tk \Tk f} \\
\Tk B & \rTo_{\delta_B} & \Tk \Tk B
\end{diagram}
\end{proposition}
\begin{proof}
This is a straightforward unfolding of the definitions. Starting with a play $\gplay \in \Tk A$, going either way around the diagram results in $[ (p_1, t_1), \ldots , (p_n, t_n) ]$, where $t_i = [ (p_1, f(a_1)), \ldots , (p_i, f(a_i))]$,  $i=1,\ldots ,n$.
\end{proof}
The natural transformation $\varepsilon$ is the \emph{counit} of the comonad, while $\delta$ is the \emph{comultiplication}.

We can now gather these elements together to complete our construction.
\begin{theorem}
The triple $(\Tk, \varepsilon, \delta)$ forms a comonad on the category $\CS$.
\end{theorem}
\begin{proof}
The remaining points to be verified from the definition of a comonad are that the following diagrams commute, for all objects $A$ of $\CS$.
\[ \begin{diagram}
\Tk A & \rTo^{\delta_A} & \Tk \Tk A \\
\dTo^{\delta_A} & & \dTo_{\Tk \delta_A} \\
\Tk \Tk A & \rTo_{\delta_{\Tk A}} & \Tk \Tk \Tk A
\end{diagram} \qquad \qquad
\begin{diagram}
\Tk A & \rTo^{\delta_A} & \Tk \Tk A \\
\dTo^{\delta_A} & \rdEq & \dTo_{\Tk \epA} \\
\Tk \Tk A & \rTo_{\epT} & \Tk A
\end{diagram}
\]
While these diagrams look somewhat formidable, verification that they commute again reduces to a straightforward diagram chase.
For the first diagram, starting with a play $\gplay \in \Tk A$, either way around the diagram takes us to $[ (p_1, t_1), \ldots , (p_n, t_n) ]$, where
$s_i = [ (p_1, a_1), \dots , (p_i, a_i) ]$, $t_i = [(p_1, s_1), \ldots , (p_i, s_i) ]$, $i=1,\ldots n$.
\end{proof}

\subsection{The Co-Kleisli Category}
\label{coK}

We now turn to one of the fundamental constructions associated with a comonad, the co-Kleisli category \cite{kleisli1965every}. We use the  notation $\CK$ for this category. The objects are the same as those of $\CS$, while a morphism from $A$ to $B$  in $\CK$ is a $\CS$-morphism $f : \Tk A \rTo B$. Given morphisms $f : \Tk A \rTo B$ and $g : \Tk B \rTo C$, we use the comonad structure to compose them:
\begin{diagram}
\Tk A & \rTo^{\delta_A} & \Tk \Tk A & \rTo^{\Tk f} & \Tk B & \rTo^{g} & C
\end{diagram}
The identity morphisms are given by the counit of the comonad: 
\[ \epA : \Tk A \rTo A . \]
We write $A \arrk B$ if there exists a morphism from $A$ to $B$ in $\CK$.

\subsubsection*{The Kleisli coextension}
\label{coexsec}

The operation $f \mapsto \Tk f \, \circ \, \delta_A$ which sends $f : \Tk A \rTo B$ to $f^* : \Tk A \rTo \Tk B$ is known as the \emph{Kleisli coextension}.
Comonads have an alternative presentation in terms of this operation and the counit maps \cite{moggi1991notions}.
For our purposes, it will be useful to have a concrete description of this operation. Given a co-Kleisli morphism $f : \Tk A \rTo B$,
\[ f^* : \gplay \mapsto [ (p_1, f(s_1)), \ldots , (p_n, f(s_n))] \]
where $s_i = [ (p_1, a_1), \dots , (p_i, a_i) ]$, $i=1,\ldots n$. 

\subsection{Grading}
\label{gradsec}

We have defined a comonad $\Tk$ for each positive integer $k$. We now consider how these are related.

We can think of the morphisms $f : A \rTo B$ in the co-Kleisli category for $\Tk$ as those which only have to respect the $k$-local structure of $A$. The lower the value of $k$, the less information available to Spoiler, and the easier it is for Duplicator to have a winning strategy. Equivalently by Theorem~\ref{strmorth}, the easier it is to have a homomorphism $\Tk A \rTo B$, \ie a morphism from $A$ to $B$ in the co-Kleisli category. This leads to a natural weakening principle: if we have a morphism from $\Tk A$ to $B$, then this should yield a morphism from $\Tl A$ to $B$ when $l < k$.

This idea is directly captured in our construction. There is an inclusion $i^{l,k}_A : \Tl A \rinc \Tk A$ whenever $l \leq k$. 
The following is easily verified.
\begin{proposition}
The inclusion maps form a natural transformation $i^{l,k} : \Tl \rTo^{\cdot} \Tk$ which is 
a morphism of comonads, \ie it preserves the counit and comultiplication.
\end{proposition}
This supports the weakening principle mentioned above. Given a morphism $f : \Tk A \rTo B$, we have a morphism
\begin{diagram}
\Tl A & \rTo^{i^{l,k}_{A}} & \Tk A & \rTo^{f} B
\end{diagram}

\subsection{Coalgebras}
Another fundamental aspect of comonads is that they have an associated notion of \emph{coalgebra}. A coalgebra for $\Tk$ is a morphism $\alpha : A \rTo \Tk A$ such that the following diagrams commute:
\[
\begin{diagram}
A & \rTo^{\alpha} & \Tk A \\
\dTo^{\alpha} & & \dTo_{\delta_A} \\
\Tk A & \rTo_{\Tk \alpha} & \Tk \Tk A
\end{diagram}  \qquad \qquad
\begin{diagram}
A & \rTo^{\alpha} & \Tk A \\
& \rdTo_{\id_A} & \dTo_{\epA} \\
& & A
\end{diagram}
\]
Note in particular that a coalgebra structure on $A$ makes it a retract of $\Tk A$ via the counit $\epA$.

A morphism of coalgebras from $(A, \alpha)$ to $(B, \beta)$ is a morphism $h : A \rTo B$ such that the following diagram commutes:
\begin{diagram}
A & \rTo^{\alpha} & \Tk A \\
\dTo^{h} & & \dTo_{\Tk h} \\
B & \rTo_{\beta} & \Tk B
\end{diagram}

Coalgebras and their morphisms form a category $\CS^{\Tk}$, the \emph{Eilenberg-Moore category} for the comonad $\Tk$. This provides another way  of looking at the co-Kleisli category. We can think of the objects in it as the \emph{cofree coalgebras} $\Tk A$, with structure maps given by comultiplication. 
Note that the diagrams for a comonad instantiate those for a coalgebra when we take the coalgebra to be $(\Tk A, \delta_A)$.
The morphisms are then taken to be the coalgebra morphisms $h : \Tk A \rTo \Tk B$. This is equivalent to the usual presentation, since we can pass from $f : \Tk A \rTo B$ to its Kleisli coextension $f^* = Tf \circ \delta_A : \Tk A \rTo \Tk B$, and from a coalgebra morphism $h : \Tk A \rTo \Tk B$ to $\epA \circ h : \Tk A \rTo B$, and these two passages are mutually inverse. This representation displays the co-Kleisli category explicitly as a full subcategory of the Eilenberg-Moore category.

Note that a coalgebra structure $\alpha$ on $A$ implies that a homomorphism exists from $A$ to $B$ whenever a homomorphism exists from $\Tk A$ to $B$. Given $h : \Tk A \rTo B$, we can form $h \circ \alpha : A \rTo B$. Thus we should only expect a coalgebra structure to exist when the $k$-local information on $A$ is sufficient to determine the structure of $A$.

Given  a  structure $A$,  we define a \emph{$k$-traversal} of $A$ to be a structure $(A, {\leq}, i)$, where $\leq$ is a partial order on $A$ which is a tree order, \ie for each $a \in A$, the predecessors of $a$  in the order form a chain; and $i : A \rarr [k]$ is a labelling map such that, whenever $a$ is adjacent to $b$ in the Gaifman graph of $A$, $a$ is comparable to $b$, say $a \leq b$, and for all $c$ such that $a < c \leq b$, $i(a) \neq i(c)$.

\begin{theorem}
\label{coalgthm}
Let $A$ be a finite structure. There is a bijective correspondence between
\begin{enumerate}
\item coalgebras $\alpha : A \rTo \Tk A$
\item $k$-traversals $(A, {\leq}, i)$.
\end{enumerate}
\end{theorem}
\begin{proof}
Let $\alpha : A \rTo \Tk A$ be a coalgebra. Let $T \subseteq \Tk A$ be the image of $\alpha$.
Unpacking the content of the two commutative diagrams for a coalgebra we see that if $\alpha(a) = [ (p_1, a_1), \ldots , (p_l, a_l) ]$, then $a_l = a$, and 
$\alpha(a_i) =  [ (p_1, a_1), \ldots , (p_i, a_i) ]$, $i=1,\ldots , l$. Thus $T$ is a prefix-closed subset of $\Tk A$, and moreover, for each $a \in A$, there is a unique $s \in T$ with last move of the form $(p, a)$. We can then define $a \leq b \ifdef \alpha(a) \preford \alpha(b)$, and $i(a) = p$ where $s[(p, a)] \in T$. Since $\alpha$ is a homomorphism, it must be the case that when $a$ and $b$ are adjacent in the Gaifman graph, $\alpha(a)$ is comparable with $\alpha(b)$, say $\alpha(a) \preford \alpha(b)$, and the pebble index $i(a)$ must not occur in the suffix of $\alpha(a)$ in $\alpha(b)$. Thus $(A, {\leq}, i)$ is a $k$-traversal of $A$.

Conversely, let $(A, {\leq}, i)$ be a $k$-traversal of $A$. We define $\alpha(a)$ by induction on the number of strict predecessors of $a$: $\alpha(a) = s[(i(a), a)]$, where
$s = \alpha(a^{-})$ if $a^{-}$ is the immediate predecessor of $a$, and otherwise $s = [\, ]$. The $k$-traversal conditions imply that $\alpha : A \rTo \Tk A$ is a coalgebra.
It is easy to see that these passages between coalgebras and $k$-traversals are mutually inverse.
\end{proof}

In the case that $k \geq n$, where $n$ is the cardinality of $A$, there is a trivial $k$-traversal of $A$ obtained by choosing a linear order $a_1 < \ldots < a_n$ on $A$ and defining $i(a_j) = j$.

We shall return to the issue of when coalgebras exist in our discussion of treewidth in Section~\ref{twsec}.

\section{Finite and infinite}
\label{nogosec}

Our primary focus is on finite structures. However, $\Tk A$ is always infinite. Is this necessary?
As we shall see in the next section, given finite structures $A$ and $B$, the question of whether there is a homomorphism $\Tk A \rarr B$ can be finitized, using a positional representation. However, to give a comonadic representation of $k$-locality, we have to define a structure on a given $A$ which will allow us to characterize the situation for \emph{all} choices of target $B$. Could this be done using a finite representation instead of $\Tk A$ with the same effect?
The following no-go result says that this is not the case, and
therefore an infinite representation cannot be avoided.  We prove it
for the special case $k=2$. 
\begin{theorem}\label{thm:nogo2}
There is no construction $A \mapsto Q A$ on finite structures $A$ such that $Q A$ is finite, and for all finite $B$:
\[ Q A \rarr B \IFF A \rightarrow_2 B . \]
\end{theorem}
\begin{proof}
We assume that $\sigma$ contains a binary relation symbol $E$ and let
$A$ be the structure with three elements $a,b,c$ where $E$ is
interpreted as the three-element cycle, i.e.\ $E(a,b),E(b,c),E(c,a)$
hold and no other pairs are related by $E$. All other relations in $\sigma$ are interpreted by the empty relation on $A$.

We claim that for any finite $B$, $A \rightarrow_2 B$ if and only if
$B$ contains an $E$-cycle.  In one direction, we note that $\mathbb{T}_2 A$
contains an infinite $E$-path: $[(1,a)]$, $[(1,a),(2,b)]$,
$[(1,a),(2,b),(1,c)]$, $[(1,a),(2,b),(1,c),(2,a)], \ldots$.  The
homomorphic image of an infinite path inside a finite structure must
contain a cycle.  In the other direction, assume that $B$ contains an
$E$-cyle $C$.  We define a map $h: \mathbb{T}_2 A \rightarrow B$ by induction on
the length of plays in $\mathbb{T}_2$, such that the image of $h$ is contained
in the cycle $C$.  For $s = [(p,x)]$ where $p \in \{1,2\}$
and $x \in \{a,b,c\}$, choose $h(s)$ to be an arbitrary element of the
cycle in $B$.  Suppose now that $s$ has length at least 2, the last
move in $s$ is $(p,x)$ and by induction $h(s')$ has been defined for
all proper prefixes $s'$ of $s$.  Let $t$ be the longest prefix of $s$
ending in a move $(q,y)$ for $q \neq p$.  We define $h(s)$ according
to the three cases: if $x=y$ then $h(s) = h(t)$; if $E(x,y)$ then let
$h(x)$ be the element of $C$ with an $E$-edge to $h(t)$; and if
$E(y,x)$, then let $h(x)$ be the element of $C$ with an $E$-edge from
$h(t)$.  It is then easily checked that $h$ is a homomorphism.

Now, suppose there was a finite $Q A$ as in the statement. Then, since
$A \rightarrow_2 QA$, $QA$  contains an $E$-cycle.  Let
$m$ be the length of the shortest cycle in $QA$.  Consider the structure
$C_{m+1}$ consisting of a single directed cycle of length $m+1$.
Since it contains a cycle, $A \rightarrow_2 C_{m+1}$.  However, it is not the case
that $QA \rightarrow C_{m+1}$, since the homomorphic image of a cycle of length $m$
must be a closed walk of length $m/l$ for some $l$, and $C_{m+1}$ contains no such
walk.
\end{proof}

We can apply this result to rule out a categorical formulation of our question.
\begin{corollary}
There is no comonad defined on $\CSf$ whose co-Kleisli category has the same preorder collapse as $\Ktwo$.
\end{corollary}
This version of the no-go theorem is stated in purely categorical terms, while the proof uses very characteristic finite model theory arguments --- indeed, essentially this argument appeared in the proof of Proposition~7.9 in \cite{atserias2006preservation}.

We conjecture that this result extends to all higher values of $k$.  
Indeed, a
 generalization of the construction in the proof above can be used to
 show that for each $k$ there is a signature $\sigma$, containing a
 $k$-ary relation symbol, and a $\sigma$-structure $A$, for which there is no
 finite $Q A$ with $Q A \rarr B$ iff  $A \rightarrow_k B$. It remains to prove the result uniformly for all signatures.
 
\section{Positions and strategies}
\label{possec}

We now turn to the situation where finite structures $A$ and $B$ are given.
We define $\Gmk(A, B)$, the set of \emph{configurations} or \emph{positions} for the existential $k$-pebble game from $A$ to $B$, to be the set of all partial functions $\gm : [k] \rpfn A \times B$.
The idea is that $\gm$ represents the position left on the board after some rounds of the game. The domain of $\gm$ is the set of pebbles used so far; $\gm(p) = (a, b)$ means that Spoiler  currently has pebble $p$ placed  on $a \in A$, while Duplicator has her matching pebble on $b \in B$. We include the empty partial function, which represents the initial configuration.

We shall use the update operation $\gm[p\mapsto (a, b)]$ on configurations. This yields the configuration with domain $\dom(\gm) \cup \{ p \}$, and such that
\[ \gm[p\mapsto (a, b)](q) = \left\{ \begin{array}{lr}
(a, b), & p = q \\
\gm(q) & p \neq q
\end{array} \right.
\]
There are really two cases which are covered by this update operation, both of which have natural readings in terms of the pebble game. If $p \in \dom(\gm)$, then the update represents Spoiler moving pebble $p$ from its current position to mark $a$, while Duplicator moves her matching pebble to $b$. If $p \not\in \dom(\gm)$, then we are extending the domain of $\gm$, which corresponds to Spoiler placing a previously unused pebble, and Duplicator her matching pebble.

It will be convenient to use the transition notation $\gm \rTo{(p,a):b} \gm'$, where $\gm' = \gm[p\mapsto (a, b)]$; and $\gm \rTo \gm'$ if $\gm \rTo{(p,a):b} \gm'$ for some $p, a, b$.

We now relate positions to plays. A strategy for Duplicator in the existential $k$-pebble game from $A$ to $B$ can be represented by a function $f: \Tk A \rTo B$, which responds to each move of Spoiler, in the context of the previous history of the game, with a move in $B$. The coextension $f^*$, defined as in Section~\ref{coexsec}, makes explicit the sequence of responses by Duplicator, with the matching use of pebbles to those of Spoiler.
Now consider a pair $(s, t)$ in the graph of $f^*$, \ie $s \in \Tk A$ and $t = f^*(s)$. Note that $s$ and $t$ will be the same length, and have the same sequence of pebble indices. 
Let $p_{1}, \ldots , p_{l}$ be the pebble indices occurring in $s$, with $1 \leq l \leq k$. Let $(p_{i}, a_{i})$ be the last occurrence of a move for pebble $p_{i}$ in $s$, $(p_{i}, b_{i})$  the corresponding occurrence in $t$.
We define $\posn(s, t)$ to be the configuration $\gm$ with domain $\{p_{1}, \ldots , p_{l}\}$ and $\gm(p_{i}) = (a_i,b_i)$, $i=1,\ldots , l$.
We write $\tf : \Tk A \rTo \Gmk(A, B)$ for the map
\[ \tf(s) = \posn(s, f^*(s)) . \]
Note that this is a  map from the infinite set $\Tk A$ to the finite set $\Gmk(A, B)$. 

A \emph{strategy in positional form} is given by a set $S \subseteq \Gmk(A, B)$ satisfying the following conditions:
\begin{itemize}
\item[(S0)] $\emp \in S$
\item[(S1)] For all  $\gamma \in S$, $p \in [k]$ and $a \in A$, there is $b \in B$ such that $\gm \rTo^{(p,a):b} \gm' \in S$.
\item[(S1)] $S$ is reachable: for all $\gamma \in S$, there is a sequence
\[ \gm_0 \rTo \cdots \rTo \gm_n \]
with $\gm_0 = \emp$, $\gm_n = \gm$, and $\gm_i \in S$ for all $i$.
\end{itemize}

Given a Duplicator strategy $f : \Tk A \rTo B$, we define its representation in positional form as
\[ \Gf := \{  \tf(s) \mid s \in \Tk A \} \cup \{ \emp \} , \]
the set of positions which can be reached following the strategy represented by $f$.

\begin{proposition}
\label{posrepprop}
For any strategy, the set of positions $\Gf$ is a strategy in positional form.
Conversely, for any strategy in positional form $S$, there is a function $f : \Tk A \rTo B$ such that $\Gf = S$. 
\end{proposition}
\begin{proof}
Given a strategy $f$, (S0) holds by definition of $\Gf$. For closure under (S1), suppose that 
\[ \gamma = \tf(s) . \]
If Spoiler places pebble $p$ on $a$, this leads to an extended play $s[(p, a)]$; for any such move, Duplicator has a response $b = f(s[(p, a)])$. Then 
\[ \gm \rTo^{(p,a):b} \gm' = \tf(s[(p, a)]) \in \Gf . \]
Reachability of $\gm = \tf(s)$ holds by induction on the length of $s$, since $\tf(s) \rTo \tf(s[(p,a)])$.

For the converse, given $S$ we can define $f(s)$ by induction on $|s|$. We choose a linear order on $B$.
Consider a play $s[(p,a)]$. By induction, $f(s)$ has already been defined, with $\gm = \tf(s) \in S$.
We define $X := \{ b \in B \mid \gm \rTo^{(p,a):b} \gm' \in S \}$,
\[ Y := \{ b \in X \mid \exists t = t'[(p,a)] \preford s. \, \tf(t') = \gm \wedge f(t) = b \} . \]
Then we define
\[ f(s[(p,a)]) = \left\{ \begin{array}{ll}
\min_{B} (X \setminus Y), & X \setminus Y \neq \emp \\
\min_{B} (X), & \mbox{otherwise.}
\end{array} \right.
\]
Clearly $\tf(s[(p,a)]) \in S$.
Note that by (S0), this definition also covers the case when $s$ is empty.

It remains to show that for all $\gm \in S$,  $\gm \in \Gf$. 
We argue by induction on the length of the shortest transition sequence from the empty configuration to $\gm$.
Consider $\gm \rTo^{(p,a):b} \gm'$ with $\gm = \tf(s)$, where $s$ is taken minimal in the prefix order in $\tf^{-1}(\gm)$. Let $X$ be the set specified above in the definition of $f$,
and let $b$ have $i$ strict predecessors in $X$ in the chosen linear order on $B$. Now consider the play $s' = s\underbrace{[(p,a), \ldots , (p,a)]}_{i+1}$. By definition of $f$, $f(s') = b$, so $\tf(s') = \gm'$.
\end{proof}

The reason for the somewhat involved construction in the second part of the proof is that we have to construct a deterministic strategy at the level of plays which maps onto a possibly non-deterministic strategy at the level of positions. We shall return to this point in the sequel.

\subsection{Winning conditions}

A configuration $\gamma$ is \emph{winning for Duplicator} if the relation 
\[ \rel(\gm) \, := \, \{ \gamma(p) \mid p \in \dom(\gamma) \} \subseteq A \times B \]
is a partial homomorphism from $A$ to $B$.

A strategy in positional form $S$ is winning if $\gm$ is winning for all $\gm \in S$.
A Duplicator strategy $f : \Tk A \rTo B$ is winning if $\Gf$ is winning.

It will be convenient to make use of the following device. We consider the expansion of our relational signature $\Sig$ with an additional binary relation $I$.
We turn $\Sig$-structures into $\Sig \cup \{I\}$-structures by interpreting $I$ as the identity relation. We refer to such structures as $I$-structures. Note that if we interpret $\Tk A$ over this expanded signature, it will not be an $I$-structure; however, the interpretation of the $I$ relation on $\Tk A$ where $A$ is an $I$-structure will  hold when multiple pebbles have been placed by Spoiler on the same element of $A$. This will ensure that a homomorphism $\Tk A \rarr B$ will induce a single-valued mapping on the underlying positions.

\begin{proposition}
\label{homwinprop}
Given $I$-structures $A$ and $B$, and a function $f : \Tk A \rTo B$, the following are equivalent:
\begin{enumerate}
\item $f$ is a winning strategy for Duplicator.
\item $f$ is a homomorphism.
\end{enumerate}
\end{proposition}
\begin{proof}
If $f$ is a homomorphism, the relation $\rel(\gm)$ arising from any position in $\Gf$ is single-valued, since $A$ and $B$ are $I$-structures, so multiple pebbles paced on the same element of $A$ must be mapped to the same element of $B$.
From the definition of the relational structure on $\Tk A$, preservation of the relation instances in a play $s$ is easily seen to be equivalent to the partial homomorphism condition on $\gamma = \tf(s)$.
\end{proof}

\subsection{Determinization}

We can view a strategy in positional form $S$ as a \emph{finite-state transducer}.
The set of states is $S$, with initial state $\emp$. The input set is $[k] \times A$, while the output set is $B$. For each configuration $\gamma \in S$, there are transitions
\[ \gm \rTo^{(p,a):b} \gm' . \]
Note that this transducer need not be deterministic. In general, functions $f : \Tk A \rTo B$ will give rise to non-deterministic transducers $S_f$, since different plays mapping to the same position can give rise to different transitions for given Spoiler moves $(p,a)$. However, there is a simple determinization procedure.
Given $S$, we choose a function $\tha : [k]\times A \rTo B$ for each $\gm \in S$ such that for all $(p,a) \in [k]\times A$, $\gm \rTo^{(p,a):\tha(p,a)} \gm' \in S$, and then define the reachable set of configurations $\Sr$ under this choice of transitions.
Explicitly, $\Sr$ is defined as the least fixpoint of the following monotone function on sets of transitions:
\[ \Phi(U) \; = \; \left\{  \begin{array}{l}
\{ \emp \} \; \cup \\
\{ \gm' \mid \exists \gm \in U, (p, a) \in [k] \times A. \, \gm \rTo^{(p,a):\tha(p,a)} \gm' \} 
\end{array} 
\right. \]
The following is easily verified.
\begin{proposition}
\label{detprop}
\begin{enumerate}
\item $\Sr$ is a strategy in positional form. 
\item It is deterministic: for each $\gamma \in \Sr$ and $(p,a) \in [k]\times A$, there are unique $b \in B$, $\gm' \in \Sr$ such that $\gm \rTo^{(p,a):b} \gm'$.
\item If $S$ is winning, so is $\Sr \subseteq S$.
\end{enumerate}
\end{proposition}

For finite structures $A$, $B$, we write $A \arrkd B$  if there is a co-Kleisli morphism realized by a deterministic finite-state strategy in positional form.
\begin{proposition}
\label{dfsprop}
For finite structures $A$, $B$, $A \arrk B \IFF A \arrkd B$.
\end{proposition}
\begin{proof}
Given $f : \Tk A \rTo B$, by Proposition~\ref{posrepprop} we can pass to the positional representation $S_f$, then by Proposition~\ref{detprop} determinize $S_f$ to obtain $S_d$, then, by Proposition~\ref{posrepprop} again, pass back to a function $f_d : \Tk A \rTo B$ such that $S_{f_d} = S_d$. Moreover, by Proposition~\ref{homwinprop}, if $f$ is a homomorphism, so is $f_d$. 
\end{proof}

\subsection{Co-Kleisli morphisms and winning strategies}

The following result, which is a corollary to Propositions~\ref{posrepprop}, \ref{homwinprop},  and~\ref{dfsprop},  justifies our claim that the pebbling comonad captures the content of the existential $k$-pebble game.

\begin{theorem}
\label{strmorth}
Given $I$-structures $A$ and $B$, the following are equivalent:
\begin{enumerate}
\item There is a winning strategy for Duplicator in the existential $k$-pebble game from $A$ to $B$.
\item $A \arrk B$.
\item $A \arrkd B$.
\end{enumerate}
\end{theorem}

\section{Equivalences}

Various notions of equivalence between relational structures play an important r\^ole in Finite Model Theory,
in particular the elementary equivalences induced by various logics. If $\Log$ is a logic, the corresponding equivalence is denoted $\equiv^{\Log}$, where
\[ A \eqL B \ifdef \forall \form \in \Log. \, A \models \form \iff B \models \form . \]
These equivalences can be characterized by various combinatorial games.

Our aim in this section is to characterize three important such equivalences in terms of morphisms in the co-Kleisli category.
We introduced the relation $A \arrk B$ on  structures in Section~\ref{coK}. This relation is clearly reflexive and transitive. The corresponding equivalence relation is
\[ A \rlk B \ifdef A \arrk B \; \wedge \; B \arrk A . \]
We write $\eqk$ for the elementary equivalence induced by the existential-positive fragment of $\Lk$.
\begin{proposition}
\label{eqkprop}
For all $A$, $B$: $A \rlk B \ifaoif A \eqk B$.
\end{proposition}
\begin{proof}
This is an immediate consequence of Proposition~\ref{pebbleposprop} and Theorem~\ref{strmorth}.
\end{proof}

\subsection{Counting logic equivalence}

A more interesting question is posed by isomorphism in the co-Kleisli category, which we denote by $\isoK$. What does the equivalence this induces correspond to in logical terms?
Given a morphism $f : \Tk A \rTo B$, for each $s \in \Tk A \cup \{ [\,]\}$ and $p \in [k]$, there is a function $\thsp : A \rTo B$ such that, for all $a \in A$, $f(s(p, a)) = \thsp(a)$.
If these functions are all injective, respectively surjective, we write $A \arrki B$, respectively $A \arrks B$.
The corresponding equivalences are denoted $A \rli B$, $A \rls B$. If there is a morphism $f : \Tk A \rTo B$ such that the functions $\thsp$ are all bijective, and moreover 
for all $s \in \Tk A$ with $\gm = \posn(s, f^*(s))$, the relation $\rel(\gm)$ is a  partial isomorphism, we write $A \arrkb B$.

It is standard that for finite structures $A$ and $B$, if there are injective homomorphisms $A \rightarrow B$ and $B \rightarrow A$, then $A \cong B$; and similarly for surjective homomorphisms. Localizing these arguments to the functions in context $\thsp$ yields the following result.
\begin{theorem}
\label{eqeqth}
For all  finite $A$, $B$:
\[ A \rli B  \ifaoif A \rls B \ifaoif A \arrkb B \ifaoif A \isoK B . \]
\end{theorem}
\begin{proof}
Since $A$ and $B$ are finite, if there exist injective maps from $A$
to $B$ and $B$ to $A$, the two sets have the same number of elements.
It follows that any injective map is in fact surjective, and the first
equivalence is immediate.  Now, fix morphisms $f: \Tk A \rTo B$ and
$g: \Tk B \rTo A$ witnessing $A \arrki B $ and $B \arrki A$
respectively.  Then $f^*$ and $g^*$ are both injective maps.  
Consider the finite substructures $P_A$ and $P_B$ of $\Tk A$ and $\Tk B$ respectively induced by sequences of a length at most $n$.  It is easily seen that $f^*$ and $g^*$ restricted to these substructures are injective homomorphisms  $P_A \rarr P_B$ and $P_B \rarr P_A$ respectively.  It follows that they are isomorphisms of these finite structures, and so
$\rel(\gm)$ is a  partial isomorphism when  $\gm = \posn(s, f^*(s))$
for any $s$.  Thus, $f$ witnesses $A \arrkb B $.

For the final equivalence, suppose we are given a homomorphism 
$f:  \Tk A \rTo B$ witnessing $A \arrkb B $.  Then, a simple induction
on the length of the sequences establishes that $f^*$ is a bijection
between $\Tk A$ and $\Tk B$ and indeed is an isomorphism.  Therefore
$A \isoK B$.  Conversely, if $h:  \Tk A \rTo \Tk B$ is an isomorphism,
we let $f:  \Tk A \rTo B$ be given by $f = \epB \circ h$ and note that
this satisfies the condition that all functions $\thsp$ are bijective. Moreover, if $t = h^*(s)$, then $s = k^*(t)$ where $k$ is the inverse of $h$ in $\CK$. The fact that $h$ and $k$ are both morphisms implies that $\rel(\gm)$ is a partial isomorphism, where $\gm = \posn(s, t)$.
\end{proof}

We now recall the \emph{bijection game} from \cite{hella1996logical}. This is a variant of the pebble game, in which at each round
with corresponding position $\gm$, Spoiler selects a pebble $p$, and Duplicator then selects a bijection $h$ between $A$ and $B$ such that, for all $q \neq p$ with $\gm(q) = (a,b)$, $h(a) = b$. Duplicator wins the round if, for each $a \in A$, $\rel(\gm[p\mapsto a])$ is a partial isomorphism. 

The bijection game characterizes the equivalence induced by the
counting logic $C^k$ \cite{immerman1986relational}, the $k$-variable
logic with counting quantifiers, which plays a central r\^ole in
finite model theory.
\begin{theorem}[\cite{hella1996logical}]
For all finite $A$, $B$, the following are equivalent:
\begin{enumerate}
\item $A \equiv^{C^{k}} B$
\item Duplicator has a winning strategy in the bijection game.
\end{enumerate}
\end{theorem}

The following is essentially immediate from the definitions:
\begin{proposition}
For all $A$, $B$, Duplicator has a winning strategy in the bijection game if and only if $A \arrkb B$.
\end{proposition}

Combining these results with Theorem~\ref{eqeqth}, we obtain the main result of this subsection:
\begin{theorem}
For all finite $A$, $B$:
\[ A \isoK B \IFF A \equiv^{C^{k}} B . \]
\end{theorem}
\noindent Thus isomorphism in the co-Kleisli category for the pebbling comonad $\Tk$ characterizes the elementary equivalence induced by $k$-variable counting logic.

\subsection{Back-and-forth equivalence}
\label{bafsec}

Finally, we turn to the elementary equivalence induced by the full $k$-variable logic $\Lk$, which we write as $\equiv^k$. The standard game characterization of this uses back-and-forth pebble games, in which Spoiler can play in $B$ as well as $A$, and Duplicator has to respond in the other structure \cite{immerman1982upper}.
This can be defined concisely using our present notation as follows. Given $\gm \in \Gmk(A,B)$, define $\rev{\gm} \in \Gmk(B,A)$ with $\rev{\gm}(p) = (b,a)$ iff $\gm(p) = (a,b)$; and given $S \subseteq \Gmk(A)$, define $\rev{S} = \{ \rev{\gm} \mid \gm \in S \}$. Now a positional-form winning  strategy for Duplicator in the $k$-pebble game between $A$ and $B$ is $S \subseteq \Gmk(A,B)$ such that $S$ is a winning strategy in the existential $k$-pebble game from $A$ to $B$, and $\rev{S}$ is a winning strategy in the existential $k$-pebble game from $B$ to $A$. Spelling this out, we see that this requires that $S$ satisfies a back-condition as well as the usual forth-condition as in the existential case; and that $R(\gm)$ is a partial isomorphism for all $\gm \in S$.
\begin{theorem}[\cite{immerman1982upper}]
\label{pbthm}
For all finite structures $A$ and $B$, there is a winning strategy for Duplicator in the $k$-pebble game between $A$ and $B$ iff $A \equiv^{k} B$.
\end{theorem}

We now define the relation $A \rlkm B$ if there are co-Kleisli arrows $f : \Tk A \rTo B$ and $g : \Tk B \rTo A$ such that $\rev{S_f} = S_g$.
\begin{theorem}
For all finite structures $A$ and $B$:
\[ A \rlkm B \IFF A \equiv^{k} B . \]
\end{theorem}
\begin{proof}
The forward implication follows directly from Theorem~\ref{pbthm}. For the converse, given winning existential strategies $S$ from $A$ to $B$ and $\rev{S}$ from $B$ to $A$, we obtain corresponding $f$ and $g$ from Theorem~\ref{posrepprop}.
\end{proof}

An interesting point arising from this result is the necessity for non-deterministic positional strategies. While existential strategies witnessing homomorphisms can always be determinized by Proposition~\ref{detprop}, in general a coupled pair of strategies $(S, \rev{S})$ cannot both be deterministic. A simple example where this arises is given by the complete graphs $K_k$ and $K_{k+1}$. Note that $K_k \equiv^k K_{k+1}$, while $K_k {\not\equiv}^{C^k} K_{k+1}$.

\section{Coalgebra number and treewidth}
\label{twsec}

The notion of treewidth of a graph \cite{robertson1986graph}, extended
to general relational structures in \cite{feder1998computational},
plays a pervasive r\^ole in identifying ``islands of tractability'' in
algorithmic graph theory. We shall write $\tw(A)$ for the treewidth of a structure $A$.

We now consider how the  comonadic structure of $k$-pebbling can be used to characterize treewidth.

\begin{theorem}
For all structures $A$, $\tw(\Tk A) < k$.
\end{theorem}
\begin{proof}
We shall define a tree decomposition $(T, a)$ of $\Tk A$. The set of nodes of $T$ is $\Tk A \, \cup \, \{ [\, ]\}$, the set of plays together with the empty sequence.
We say that $s$ is adjacent to $t$ if $s \prd t$ or $t \prd s$, where $\prd$ is the immediate predecessor relation induced by the prefix order; thus $s \prd t$ iff $t = s[(p, a)]$ for some $p$, $a$. The unique path between any $s$ and $t$ goes from $s$ via $\succ$-instances of the adjacency relation to $u$, and then by $\prd$-instances from $u$  to $t$, where $u$ is the greatest common prefix of $s$ and $t$. The labelling function $a(s)$ assigns the set of \emph{active prefixes} of $s$ to $s$, where $t$ is an active prefix of $s$ if $t \preford s$, and the pebble index used in the last move in $t$ does not occur in the suffix of $t$ in $s$. Clearly, the maximum size of any $a(s)$ is $k$.

We now verify the conditions for $(T, a)$ to be a tree decomposition of $\Tk A$. Firstly, $s \in a(s)$. Secondly, if $R^{\Tk A}(s_1, \ldots , s_n)$, then the $s_i$ must all be active prefixes of some $u$. Finally, suppose that $s$ is in $a(t_1) \cap a(t_2)$. Let $u$ be the largest common prefix of $t_1$ and $t_2$. Then we must have $s \preford u$, and moreover $s$ active in $t_1$ implies that $s$ is active in all $v$ with $u \preford v \preford t_1$, and similarly for $t_2$. Thus $s \in a(v)$ for all $v$ in the unique path from $t_1$ to $t_2$ in $T$.
\end{proof}

Thus, although $\Tk A$ is infinite, it has bounded treewidth.
As an immediate consequence of this result, we have:
\begin{proposition}
\label{injhomprop}
If there is an injective homomorphism $A \rarr \Tk B$, then $\tw(A) < k$.
\end{proposition}
\begin{proof}
If there is an injective homomorphism from $A$ to $\Tk B$, then $A$ is isomorphic to a substructure of $\Tk B$, hence $\tw(A) \leq \tw(\Tk B) < k$.
\end{proof}

In particular, if there is a coalgebra $\alpha : A \rTo \Tk A$, then
$\tw(A) < k$. We define the \emph{coalgebra number} $\cn(A)$ of a
finite structure $A$ to be the least $k$ such that there is a
coalgebra $\alpha : A \rTo \Tk A$. Combining the previous Proposition
and the remark after Theorem~\ref{coalgthm}, we see that $\tw(A) <
\cn(A) \leq \card(A)$.  In fact, $\cn(A)$ yields  an elegant, purely (co)algebraic \emph{characterization} of treewidth.
\begin{theorem}
For all finite structures $A$:
\[ \cn(A) = \tw(A) + 1 . \]
\end{theorem}
\begin{proof}
 We have already seen that $\tw(A) < \cn(A)$, so it suffices to show
that $\cn(A) \leq \tw(A) +1$.  Suppose that $\tw(A) = k-1$.  This
means that there is a tree decomposition $(T,a)$ where $T$ is a tree,
and for any node $t$ of $T$, $a(t)$ is a set of at most $k$ elements
of $A$.  By standard means, we can assume that $T$ is a rooted,
directed, tree and for any node $t$ of $T$, there is at most one
element of $A$ that appears in $a(t)$ that does not appear in $a(s)$
for any ancestor $s$ of $t$ (for instance, we could take $(T,a)$ to be
a \emph{nice} decomposition as defined in~\cite{kloks1994}).  

For each $x \in A$, we can then define $t(x)$ to
be the least (in the tree order) element of $T$ such that $x \in
a(t(x))$.  The properties of a tree decomposition guarantee that this
is well defined, and our assumption that $(T,a)$ is nice ensures that
$t$ is injective.  The function $t$ induces a partial order on $A$: $x \leq y$ if and only
if $t(x)$ is an ancestor of $t(y)$.  It is easily seen that this is a
tree order.  We now proceed to define a map
$i: A \rightarrow [k]$ by induction on this order.  Suppose $i(x)$ has
been defined for all $x < y$.  In particular, this means (since $t$ is
injective) that $i(x)$
is defined for all $x \in a(t(y))$ other than $y$.  Since $a(t(y))$
has at most $k$ elements, there is at least one value $p \in [k]$ such
that $i(x) \neq p$ for all  $x \in a(t(y))$ other than $y$.  We set
$i(y)$ to be the least such $p$.

It can now be verified that $(A,\leq,i)$ is a $k$-traversal.
By
Theorem~\ref{coalgthm}, this gives a coalgebra map $A \rightarrow \Tk A$, and
establishes that $\cn(A) \leq k$.

\end{proof}

\section{Cores and pebble number}

We now turn to the question of when there exists a homomorphism (not necessarily a coalgebra) $A \rarr \Tk A$. 
We define the \emph{pebble number} $\pb(A)$ of a finite structure $A$ to be the least $k$ such that $A \rarr \Tk A$. 

The core of a structure is a key notion in graph theory and finite model theory \cite{hell1992core}.
We define the core of a finite structure $A$ to be a substructure $A' \subseteq A$ which admits a homomorphism $A \rTo A'$, but for which no proper sub-structure does.
We say that a structure is a core if it is the core of itself.
The following summarizes some basic properties of cores.

\begin{proposition}
\label{coreprop}
\begin{enumerate}
\item Every finite structure has a core.
\item The core of a structure is unique up to isomorphism.
\item A substructure is a core of $A$ if and only if it is a retract of $A$, and minimal among retracts of $A$ with respect to inclusion.
\item A structure is a core if and only if it admits no proper retracts.
\item If $A'$ is the core of $A$ and $B'$ the core of $B$, then $A \rightarrow B \IFF A' \rightarrow B'$.
\end{enumerate}
\end{proposition}
By virtue of these results, we can write $\core(A)$ for the core of a finite structure $A$.

We shall make use of some results from the literature.  We denote the \emph{canonical conjunctive query} for a finite structure $A$ by $Q^A$. This is a positive existential first-order sentence. We have the following classic result of Chandra and Merlin \cite{chandra1977optimal}:
\begin{theorem}
\label{CMthm}
The following are equivalent for finite structures $A$, $B$:
\vsn
\begin{enumerate}
\item $A \rarr B$.
\item $Q^B \models Q^A$.
\item $B \models Q^A$.
\end{enumerate}
\end{theorem}

We also have the following results due to Dalmau, Kolaitis and Vardi \cite{dalmau2002constraint}. 
\begin{theorem}
\label{DKV1thm}
The following are equivalent for finite structures $A$:
\vsn
\begin{itemize}
\item $\tw(\core(A)) < k$.
\item $Q^A$ is logically equivalent to a sentence of $\lk$, the existential positive fragment of $k$-variable first-order logic.
\end{itemize}
\end{theorem}

We say that a sentence $\vphi$ of the existential positive fragment of first-order logic can be transformed to another sentence $\psi$ by \emph{standard rewriting} if there are a sequence of steps, each of which is either (1) associative-commutative rewriting of conjunctions, (2) change of bound variable in an existential quantifier, or (3) replacing $\exists x. \, (\phi_1 \wedge \phi_2)$ by $\exists x. \, \phi_1 \, \wedge \, \phi_2$, where $x$ does not occur free in $\phi_2$, which transform $\vphi$ into $\psi$.
The following result from \cite{dalmau2002constraint} is a refinement of the previous theorem.
\begin{theorem}
\label{DKV2}
The following are equivalent for finite structures $A$:
\vsn
\begin{itemize}
\item $\tw(A) < k$.
\item $Q^A$ can be transformed by standard rewriting to a sentence of $\lk$.
\end{itemize}
\end{theorem}

The following result shows how $\Tk A$ captures  the information about $A$ which is expressible in $\lk$. 
\begin{proposition}
\label{lkprop}
Let $A$ be a finite structure.
Every sentence of $\lk$ satisfied by $A$ is  also satisfied by a finite substructure of $\Tk A$.
\end{proposition}
\begin{proof}
We list the variables of $\lk$ as $x_1, \ldots , x_k$. Given a formula $\vphi(\vx)$ with a list of distinct free variables $\vx = x_{i_1}, \ldots , x_{i_l}$, with $0 \leq l \leq k$, and a corresponding semantic tuple $\va = a_1, \ldots , a_l$, we define a tuple $\vs(\vx,\va)$ with $\vs_{0} = [ \, ]$, and $\vs_{i+1} = \vs_{i}[(i+1, \va_i)]$, for $i = 0, \ldots , l -1$.

We show by induction on formulas $\vphi(\vx) \in \lk$ that $A, \va \models \vphi(\vx)$ implies that $\Tk, \vs \models \vphi$ where $\vs = \vs(\vx,\va)$. The base case is for atomic formulas $R(\vx)$. This follows immediately from the definition of the relational structure on $\Tk$. For $\exists x. \, \vphi$, 
$A, \va \models \exists x. \, \vphi$ iff for some $a$: $A, \va a \models \vphi$, which by induction hypothesis implies that  $\Tk A, \vs s \models \vphi$, iff $\Tk A, \vs \models \exists x. \, \vphi$. Finally, for a conjunction $\vphi \wedge \psi$, $A, \va\vb\vc \models \vphi \wedge \psi$ iff $A, \va\vb \models \phi(\vx,\vy)$, and $A, \va\vc \models \psi(\vx, \vz)$, where the variables common to $\vphi$ and $\psi$ are listed as $\vx$. By induction hypothesis, $\Tk A, \vs\vt \models \vphi$ and $\Tk A, \vs \vu \models \psi$.
Then $\Tk A, \vs' \models \vphi \wedge \psi$, where $\vs' = \vs(\vx\vy\vz, \va\vb\vc)$.
\end{proof}

\begin{proposition}
For all finite structures $A$, $A \eqk \Tk A$.
\end{proposition}
\begin{proof}
By Proposition~\ref{eqkprop}, we just need to exhibit homomorphisms $\Tk T_k A \rarr A$  and $\Tk A \rarr \Tk A$. For the first, we can compose the counit maps for $\Tk A$ and for $A$, and for the second, we can just take the identity. 
\end{proof}

We can use these results to obtain our characterization of the relationship between treewidth and the pebble number.
\begin{theorem}
For all finite structures $A$:
\[ \pb(A) = \tw(\core(A)) + 1 . \]
\end{theorem}
\begin{proof}
It clearly suffices to prove that:
\[ \tw(\core(A)) < k  \IFF \pb(a) \leq k . \]
Suppose firstly that $\tw(\core(A)) < k$. By Theorem~\ref{DKV1thm}, $Q^A$ is logically equivalent to a sentence of $\lk$. By Proposition~\ref{lkprop}, $Q^A$ is satisfied by a finite substructure $T$ of $\Tk A$. By Theorem~\ref{CMthm}, there is a homomorphism $A \rarr T \rinc \Tk A$.

Now suppose that there is a homomorphism $A \rarr \Tk A$. The image of this homomorphism is a finite substructure $T$ of $\Tk A$. By Theorem~\ref{CMthm}, $Q^T \models Q^{A}$.
By Proposition~\ref{injhomprop}, $\tw(T) < k$. Hence by Theorem~\ref{DKV2}, $Q^T$ can be transformed by standard rewriting to a formula in $\lk$. We can restrict the counit map $\epA : \Tk A \rarr A$ to a homomorphism $T \rarr A$. Using Theorem~\ref{CMthm} again, $Q^A \models Q^T$. Hence $Q^A$ and $Q^T$ are logically equivalent, which shows that $Q^A$ is logically equivalent to a sentence in $\lk$. By Theorem~\ref{DKV1thm}, this implies that $\tw(\core(A)) < k$.
\end{proof}

This result shows that the existence of a homomorphism from $A$ to $\Tk A$
characterizes the treewidth of the core of $A$. The
stronger hypothesis of having a coalgebra map $A \rarr \Tk A$ characterizes the treewidth of $A$ itself.

\section{Strong $k$-consistency and contextuality}

We now look at \emph{strong $k$-consistency}, which is a fundamental
notion in constraint satisfaction \cite{kolaitis2000game}. The
non-uniform CSP problem $\CSP(B)$ can be posed for a fixed finite
structure $B$ as the existence of a homomorphism $A \rarr B$ for a
given finite structure $A$. An instance $A$ is $i$-consistent if for
every homomorphism $h$ from an $(i-1)$-element substructure of $A$ to
$B$ and every $a\in A$, there is a homomorphism on a substructure of
$A$ including $\dom(h) \cup \{a\}$ extending $h$.  It is strongly $k$-consistent if it is $i$-consistent for all $i \leq k$.

This leads us to the following notion of \emph{consistency number} $\scB(A)$ of a finite structure $A$: the greatest $k$ such that there is a homomorphism $\Tk A \rarr B$, \ie $A \arrk B$.
Note that,   if we parameterize the pebble number to define $\pb_B(A)$ as the smallest $k$ such that there is a homomorphism $A \rarr \Tk B$, then $\pb_B$ and $\scB$ are formally dual. 

If $k \geq \card(A)$, then $A \arrk B$ if and only if $A \rarr B$, so we define $\scB(A)$ as the largest $k \leq n$ such that $A \arrk B$. Also, because of the graded structure of the monads $\Tk$, note that it is strong $k$-consistency which we are capturing. Once we have $A \arrk B$, then also $A \arrl B$ for $l \leq k$.

The interesting case is where $\scB(A) < \card(A)$. This is the case where we can get partial solutions for up to $k$ variables, but no global solutions.
We note that this case covers a variety of examples of \emph{contextuality} in quantum mechanics, as shown in \cite{AbramskyBrandenburger}. In particular, many proofs of the Kochen-Specker theorem \cite{KochenSpecker} provide examples of this phenomenon. As an illustration, we consider the Mermin Magic Square \cite{mermin1990simple}:
\begin{center}
\begin{tabular}{|c|c|c|}
\hline
$A$ & $B$ & $C$ \\ \hline
$D$ & $E$ & $F$ \\ \hline
$G$ & $H$ & $I$ \\ \hline
\end{tabular}
\end{center}
The constraints are that each row and the first two columns have even parity, and the final column has odd parity.
This translates into 6 linear equations over $\ZZ_2$:
\begin{center}
\begin{tabular}{lcl}
$A \oplus B \oplus C = 0$ & $\quad$ &  $A \oplus D \oplus G = 0$ \\
$D \oplus E \oplus F = 0$ & &  $B \oplus E \oplus H = 0$ \\
$G \oplus H \oplus I = 0$ & & $C \oplus F \oplus I = 1$ \\
\end{tabular}
\end{center}
Of course, these equations are not satisfiable in $\ZZ_2$. The system is $8$-consistent but not globally consistent. The significance of this construction in quantum mechanics is that we can interpret the variables with quantum observables in such a way that the specified constraints are exactly the predicted behaviour of measuring the observables. The global inconsistency corresponds to the impossibility of  a non-contextual explanation for this behaviour. For more on this topic, see \cite{AbramskyEtAl:ContextualityCohomologyAndParadox,abramsky2013robust}.

We can also make a simple connection between pebble number and consistency number.
\begin{proposition}
If $\pb(A) \leq \scB(A)$, then $A \rarr B$.
\end{proposition}
\begin{proof}
If $\pb(A) \leq \scB(A)$, then for some $k$, 
\[ A \rarr T_k A \rarr B . \]
\end{proof}

\subsection{Cai-F\"urer-Immerman construction}
Fix a relational signature with two ternary relations $R_0$ and $R_1$.  Consider the structure $Z2$ in this signature with universe $\{0,1\}$ where $R_0 = \{(i,j,k) \mid i\oplus j \oplus k = 0\}$ and $R_1 = \{(i,j,k) \mid i\oplus j \oplus k = 1\}$.  Then, a system of equations over $\ZZ_2$ (as in the example above) can be seen as a structure $A$ whose universe is the set of variables and $(v_1,v_2,v_3) \in R_0$ if $v_1 \oplus v_2 \oplus v_3 = 0$ is an equation in the system and $(v_1,v_2,v_3) \in R_1$ if $v_1 \oplus v_2 \oplus v_3 = 1$ is an equation in the system.  It is then immediate that $A \rarr Z2$ if and only if the system is solvable.  This is the classical example of a constraint satisfaction problem that is not solvable by $k$-local consistency tests (though it is polynomial-time solvable by Gaussian elimination).  Moreover, it is also known that the class of satisfiable systems of equations is not invariant under $\equiv^{C^k}$ for any $k$.  This is shown by Atserias et al.~\cite{ABD09}, based on the construction of Cai et al.~\cite{CFI92} of a polynomial-time decidable class of graphs not definable in fixed-point logic with counting.  Here, we give an account of this construction in the categorical framework we have developed.

\begin{proposition}
Suppose $A \rightarrow_k Z2$.  Then there is a pair of structures $A_0$ and $A_1$ with $A_0 \isoK A_1$, $A \rarr A_0$, and $A_1 \rarr Z2$.
\end{proposition}
\begin{proof}
We define $A_0$ and $A_1$ to both have universe $A \times \{0,1\}$. 
The interpretation of the relations in $A_0$ is given by:
$$ R_0 = \{ (a,i),(b,j),(c,k) \mid (a,b,c) \in R_0^A \land i \oplus j \oplus k = 0 \}$$
$$ R_1 = \{ (a,i),(b,j),(c,k) \mid (a,b,c) \in R_1^A \land i \oplus j \oplus k = 0 \}.$$
The interpretation of the relations in $A_1$ is given by:
$$ R_0 = \{ (a,i),(b,j),(c,k)) \mid (a,b,c) \in R_0^A \land i \oplus j \oplus k = 0 \}$$
$$ R_1 = \{ (a,i),(b,j),(c,k) \mid (a,b,c) \in R_1^A \land i \oplus j \oplus k = 1 \}.$$

The homomorphism $A \rTo A_0$ is given by $a \mapsto (a,0)$ and the homomorphism $A_1 \rTo Z2$ by $(a,i) \mapsto i$.

Now, suppose $f: \Tk A \rTo Z2$ is a homomorphism.  We aim to use this to define an isomorphism $f': \Tk A_0 \rTo \Tk A_1$.  For any $s \in \Tk A_0$, we write $\pi s$ to denote the sequence in $\Tk A$ obtained by replacing each move $(p,(a,i))$ in $s$ by $(p,a)$.  Let $s$ be the sequence $(p_1,(a_1,i_1)),\ldots,(p_n,(a_n,i_n))$ and let $s_1,\ldots,s_n$ be the sequence of its non-empty prefixes.  We then define $f'(s)$ to be the sequence in $\Tk A_1$ whose $j$th component is $(p_j,a_j,i_j \oplus f(\pi s_j))$ and verify that this is an isomorphism.
\end{proof}

\section{Modalities}

Another facet of monads and comonads is their r\^ole in categorical logic, as categorified interpretations of S4 modalities \cite{bierman2000intuitionistic}.
The pebbling comonad $\Tk$ provides a semantics for a modality $\Bk$, which we can read as: ``it is $k$-locally the case that \ldots''.
We can take e.g. first-order logic  as a base, and extend it with the formation rule that if $\vphi$ is a formula with at most $k$ free variables, then $\Bk \vphi$ is a formula.

We extend the usual definition of satisfaction of a formula
\[ A, \va \models \vphi(\vx) \]
where $A$ is a structure, and $\va$ a sequence of elements of $A$ providing interpretations for the free variables $\vx$ of $\vphi$, with the following clause:
\[ A, \va \models \Bk \vphi  \ifdef \Tk A, \alpha(\va) \models \vphi \]
where if $\va = \langle a_1, \ldots , a_l \rangle$, $l \leq k$, $\alpha(a_i) =  [ (p_1, a_1), \ldots , (p_i, a_i) ]$, $i=1,\ldots , l$.

We can read off a number of properties of this modality directly from the comonadic structure. Firstly, the S4 axioms are valid in this semantics \emph{with respect to $\lk$ formulas $\vphi$}:
\[ \begin{array}{cl}
(\mathbf{T}) & \models \Bk \vphi \imp \vphi \\
(\mathbf{4}) & \models \Bk \vphi \imp \Bk \Bk \vphi 
\end{array}
\]
We also have the grading axiom:
\[ \models \Bk \vphi \imp \Bl \vphi \qquad (k \leq l) \]

Using Proposition~\ref{lkprop}, for formulas $\vphi$ of $\lk$ we also have the following:
\[ \models \vphi \imp \Bk \vphi . \]
If we want to reason about structures with bounded treewidth, we can use the results of the previous section to introduce suitable assumptions.

We can also consider a modality $\Dh$, which can be read as ``it is homomorphically true that''.
This has the following semantics:
\[ A, \va \models \Dh \vphi  \ifdef \exists h : A \rarr B: \; B, h(\va) \models \vphi . \]
This satisfies the dual S4 axioms with respect to $\lk$ formulas:
\[ \begin{array}{cl}
(\mathbf{T}) & \models  \vphi \imp \Dh \vphi \\
(\mathbf{4}) & \models \Dh \Dh \vphi \imp \Dh \vphi 
\end{array}
\]

Whereas the logics usually considered in finite model theory are interpreted in one structure at a time, these modalities allows us to navigate around the category of structures, treating them as different ``possible worlds''.

\section{Conclusions}

While we are not aware of any closely related work, the work of Bojanczyk \cite{bojanczyk2015recognisable} and of Adamek \textit{et al} \cite{urbat2016one} on recognizable languages over monads is in a broadly kindred spirit.
The aim of these works is to use monads as a unifying notion for the many variations on the theme of recognizability.
Also, we can mention the use of comonads in functional programming \cite{uustalu2008comonadic,uustalu2005essence}. In particular, the comonadic structure of lists is studied in \cite{orchard2014programming}. The forgetful functor  from structures to sets carries the pebbling comonad to this list comonad (although concretely the correspondence is up to list reversal).

The ideas developed in this paper suggest a number of further developments. Firstly, can we find similar characterizations of other forms of games, and connections with the corresponding logical notions? 
In fact, there is a natural family of comonads corresponding to Ehrenfeucht-Fraiss\'e games, which can be used to characterize equivalence up to given quantifier rank, using ideas analogous to those in Section~\ref{bafsec}.
There is also an unravelling comonad, which arises in the analysis of bisimulation for modal and guarded logics \cite{gradel2014freedoms}.
Details of these constructions will appear in a sequel to the present paper.

There are many other natural lines for future investigation.
Can we give categorical characterizations of key combinatorial parameters, in a similar fashion to treewidth?
Can we use categorical limit and colimit constructions, possibly in an enriched form, to connect with the combinatorial developments in \cite{nesetril2013unified}?
Can we  leverage results in descriptive complexity to give categorical descriptions of  complexity classes?

Another project  is to analyze Rossman's theorem on homomorphism preservation \cite{rossman2008homomorphism} from a categorical viewpoint.
It seems that many of the technical notions developed in the proof of this result are susceptible of a more abstract formulation, which may lead to more conceptual proofs, and the possibility of finding wider applicability for the techniques developed there.

The ideas described in this paper provide a connection between two broad themes within the field of logic in computer science: the interaction of logic with the analysis of algorithms and computational complexity, and the study of the semantics of programs and processes. Examples of such connections are still fairly rare, and we hope that the ideas introduced here can be developed further in a fruitful fashion.

\medskip

\noindent
\textbf{Acknowledgement:} The work reported here was initiated at the programme on \emph{Logical Structures in Computation} at the Simons Institute for the Theory of Computing, Berkeley.

\bibliographystyle{IEEEtran} 
\bibliography{pb}
\end{document}